\documentclass[superscriptaddress]{revtex4-1}%
\usepackage{amsfonts}
\usepackage{amsmath}
\usepackage{amssymb}
\usepackage{dsfont}
\usepackage{graphicx}
\usepackage{xcolor}
\usepackage{verbatim}
\usepackage{hyperref, bbm}
\usepackage{multirow}
\providecommand{\U}[1]{\protect\rule{.1in}{.1in}}

\newcommand{\bee}{\begin{enumerate}}
\newcommand{\eee}{\end{enumerate}}
\newcommand{\bei}{\begin{itemize}}
\newcommand{\eei}{\end{itemize}}
\newcommand{\beq}{\begin{eqnarray} \begin{aligned}}
\newcommand{\eeq}{\end{aligned} \end{eqnarray} }
\newcommand{\bea}{\begin{array}}
\newcommand{\eea}{\end{array}}
\newcommand{\ben}{\begin{eqnarray}}
\newcommand{\een}{\end{eqnarray} }

\newtheorem{theorem}{Theorem}

\newtheorem{definition}[theorem]{Definition}

\newenvironment{proof}[1][Proof]{\noindent\textbf{#1.} }{\ \rule{0.5em}{0.5em}}

\def\tr{{\rm Tr}}
\def\ot{\otimes}
\def\ep{\epsilon}

\def\<{\langle}
\def\>{\rangle}

\def\hcal{{\cal H}}
\def\pcal{{\cal P}}
\def\classops{{\cal C}}
\def\useless{{\cal S}}
\def\boundset{{\cal S_i}}
\def\id{\mathds{1}}

\usepackage{xspace}

\usepackage{multibib}

\newcommand{\be}{\begin{eqnarray} \begin{aligned}}
\newcommand{\ee}{\end{aligned} \end{eqnarray} }

\begin{document}
\title{(Quantumness in the context of) Resource Theories}
\author{Micha\l~Horodecki}
\affiliation{Institute for Theoretical Physics and Astrophysics, University of Gda\'nsk,  Gda\'nsk, Poland}
\affiliation{National Quantum Information Centre of Gda\'nsk, Sopot, Poland}
\author{Jonathan~Oppenheim}
\affiliation{Department of Physics and Astronomy, University College London, London, United Kingdom}

\begin{abstract}
We review the basic idea behind resource theories, where we quantify quantum resources by specifying a restricted class of operations.
This divides the state space into various sets, including states which are free (because they can be created under the
class of operations), and those which are a resource (because they cannot be).  One can quantify
the worth of the resource by the relative entropy distance to the set of free states, and under certain conditions, this is a unique 
measure which quantifies the rate of state to state transitions. 
The framework includes entanglement, asymmetry and purity theory.  It also includes thermodynamics, which is a hybrid resource theory
combining purity theory and asymmetry.  Another hybrid resource theory which merges purity theory and
entanglement can be used to study quantumness of correlations and discord, and we present quantumness in this more 
general framework of resource theories.  
\end{abstract}
\volumeyear{year}
\volumenumber{number}
\issuenumber{number}
\eid{identifier}
\date{\today}




\maketitle

\section{Introduction}
Quantum information theory is novel in that it seeks to apply techniques from computer science to understand the laws of physics.  The types of questions one asks are often different to the questions physicists used to ask.  We are interested in finding out what is possible given some resources or some class of operations.  We explore the scope of physical law by asking what tasks the theory allows us to perform.  We place some restriction, or enlarge the types of operations one can perform, and ask what is possible.  We can even imagine modifying the laws of physics, and ask what new tasks we could perform given this modification.

This has led to the notion of a
{\it quantum resource theory}\cite{thermo-ent2002,devetak-framework}, which are specified by a class of operations $\classops$.  The operations are often a restriction on the full set of quantum operations that can be implemented. Given this restriction, it will not be possible to create all possible quantum states from some 
fixed initial state, and as a result, the states
which cannot be created under the class of operations, will be a resource, since they cannot be created for free.
There will also be a set of states $\useless$ which can be created under the class of operations $\classops$.  We call this the set of 
 {\it free states}, because if you are given one, 
it's not particularly valuable -- you could anyway create them under the class $\classops$.  
On the other hand, the states which cannot be created by means of $\classops$ then naturally acquire some value.
If the parties facing the restriction acquire such a quantum state outside the restricted set, then they can use this state to implement measurements and transformations that they could not otherwise implement under the restriction, 
consuming the state in the process. For example, with the aid of some resource states, one might be able to fully implement the full set 
of quantum operations.  Therefore, such states become a resource.  The class of operations might also induce other sets of states $\boundset$ which are closed, in the sense that if you are given states from this set, then under the class of operations, you cannot move outside this set. 
We call such sets {\it bound sets}.

Some examples serve to illustrate the idea:
{\bf the resource theory of entanglement} --
if two or more parties are restricted to communicating classically and implementing local quantum operations (LOCC), then entangled states become a resource~\cite{BBPS1996, BBPSSW1996, horodecki_quantum_2009}.  The set $\useless$ are the 
separable states, since these can be created under the class of LOCC. A pure entangled state such as a Bell state,
cannot be created by means of LOCC, but if one is given such a resource, then one can perform tasks such as teleportation which cannot be implemented using operations from within the class of LOCC.  One can also use entangled states to create other entangled states which are outside the set of separable states.  There exist other bound sets, for example, states which have positive partial transpose, since if you are given states from this set, you cannot create a state outside this set under LOCC.

A particularly simple example, and one which we'll discuss in more detail here, is {\bf the theory of purity}: if a party is restricted to preparing states that are completely mixed and performing unitary
operations, then any state that is not completely mixed, i.e. any state that has some purity, becomes a resource~\cite{horodecki_reversible_2003},
because it is impossible to create such states under the class of operations. 
The set of free states is here just the maximally mixed state, almost by definition, 
because the class of operations allows one to prepare it, but nothing else.  On the other hand, pure states are a resource which would allow you to create other states which are not maximally mixed.  It turns out that purity theory is related to both thermodynamics and quantumness of correlations, and we 
shall discuss it in more detail in Section \ref{sec:purity}.

{\bf Asymmetry theory} is the resource theory where we restrict ourselves to quantum operations that have a particular symmetry.  States that break this symmetry then become a resource~\cite{janzing_quasi-order_2003,gour_resource_2008,marvian_pure_2011}.  For example, if one can prepare energy eigenstates,
but  operations have to conserve energy i.e. respect time-translation invariance), then the free states are mixtures of energy levels.  On the other hand
coherent superpositions over energy levels cannot be created under the class of operations, and are therefor a resource.

One typically thinks of {\bf Thermodynamics} as a theory where one cannot access microscopic degrees of freedom, but only macroscopic ones (such
as energy, volume, pressure).  In fact, being able to access microscopic degrees of freedom doesn't change the structure of the theory,
as can be seen once thermodynamics is cast as a resource theory.  One can describe 
Thermodynamics is a mixture of the resource theory of purity, and the resource theory of asymmetry (in this case, energy conservation)\cite{janzing_thermodynamic_2000,thermoiid, horodecki-singleshot}.  Like purity theory, we can only 
perform unitary operations (including those that manipulate microscopic degrees of freedom), but they also must conserve energy.  Instead of
the free resource state being the maximally mixed state, we take it to be the thermal state at temperature $T$.  
Any state which is not thermal, and is thus not in equilibrium, is therefore a resource which can be used to perform tasks such as
extracting work.  One immediately sees the advantage of casting thermodynamics as a resource theory, since having cast it as such, one can apply
it to single quantum systems, or highly correlated states~\cite{horodecki-singleshot}, as we discuss in Section \ref{sec:purity}.  In such cases, one finds
that there are two free energies, not just one, which govern the theory.

Another hybrid resource theory, and a central subject of this note, is the theory of {\bf local purity}, and this theory is related to {\bf quantumness of correlations}.  We will see that it is 
a combination of purity theory and entanglement theory.  It is an interesting resource theory to consider, in part because it has many features 
which differentiate it from other resource theories we know.  We will discuss this resource theory in Section \ref{sec:deficit}, and it's relation to the
deficit\cite{OHHH2001,nlocc} and discord\cite{HendersonVedral,Zurek01-discord}, both of which are measures of quantumness of correlations.

\begin{table}
\label{tab:var}
\begin{center}
\begin{tabular}{|c|c|c|c|c|}
\hline
Theory & Class of Operations $\classops$ & Free Set $\useless$ & Monotones & Hybrid Theory\\
\hline\hline
Entanglement & LOCC & Separable states $\sigma$ & $\inf_\sigma R(\rho||\sigma)$ & \multirow{2}{*}{Quantumness}\\
\cline{1-4}
Purity & Noisy Ops ($U$, $\id / d$) &  $\id / d$ & $R(\rho||\id/d)=\log d - S(\rho)$ & \multirow{2}{*}{Thermodynamics}\\
\cline{1-4}
Asymmetry & ${\cal E}$ s.t. $T(g)\circ{\cal E}\circ T^\dagger(g)={\cal E}$ & $\sigma$ s.t. $\sigma=T(g)\sigma T^\dagger(g)$ &  $\inf_\sigma R(\rho||\sigma)$ &\\
\hline
\end{tabular}
\end{center}
\caption{As discussed in Section \ref{sec:rel}, the relative entropy distance to the free states $\useless$ is a measure of the resource, in most resource theories.  
It is defined as $\inf_{\sigma\in\useless} R(\rho|\sigma)$ with  $R(\rho|\sigma)=\tr \rho\log\rho -\tr\rho\log\sigma$.  Under some very general conditions, 
it is the unique measure in a theory, and gives the rate at  which states can be interconverted.  The table above depicts three resource theories, and
two other resource theories (Thermodynamics and Quantumness), which can be constructed by combining the three basic ones. 
}
\end{table}

We thus have a number of interesting resource theories, and they tend to have some very common properties.
The examples above demonstrate that studying a given property of quantum states as a resource is a very efficient way of coming to understand it better.  
 The approach is also rather general: rather than considering the behavior of the property of interest for some particular system with particular dynamics, one considers instead the fundamental limits that are imposed by the restriction defining the resource and \emph{the laws of quantum theory}. On the practical side, a better understanding of a given resource helps to determine how best to implement the tasks that make use of it, and, more fundamentally, such an understanding may serve to clarify what sorts of resources are even relevant for a given task. Finally, a resource theory approach can provide a framework for organizing, synthesizing, and consolidating the
results in a given field.

Typically, resource theories provide answers to questions such as:
how does one measure the quality of different resource states? Can one particular resource state be
converted to another deterministically and if so, what is the rate of interconversion. Or can resource states be converted to 
each other indeterministically, and if so, what is the probability of failure? What about if one has
access to a catalyst? One is generally interested in identifying the equivalence classes of states that are reversibly interconvertible asymptotically.

In general, there is a lot of commonality between various resource theories.  For example, for any resource theory, there is  
a quantity -- the relative entropy distance to the set of free states or to any bound set.  This generalises
the notion of the relative entropy of entanglement~\cite{PlenioVedral1998}. It gives 
a measure of how valuable a resource state is (see Table \ref{tab:var}).  
What's more,
under certain conditions, this measure is the unique measure of the worth of the resource\cite{thermo-ent2002}.  It can then be used to determine the
rate at which states can be converted into each other under the class of operations. We will review this in more detail
in Section \ref{sec:rel}.   
In the theory of quantumness of correlations,
it means that the relative entropy of quantumness (introduced in \cite{huge-delta} and rediscovered in \cite{modi-relent-discord,groisman2007quantumness}) is a measure of quantumness.

The central aim of this paper is to review the elements of resource theories, and discuss the notion of quantumness of correlations 
within this framework.  We shall first discuss the simplest resource theory, purity theory, in a bit more detail in Section \ref{sec:purity}, and
we will also discuss it's close connection with thermodynamics, and with the thermodynamics of single systems.  We shall then in Section \ref{sec:rel}
 discuss some general features of resource theories, for example, the central role played by the relative entropy distance, and the notion of
expanded classes of operations.  In Section \ref{sec:deficit}, we review measures of quantumness of correlations, in particular the deficit,
which we then show is related to the resource theory of local purity in Section \ref{sec:def-resource}.  In the Appendix we prove the central theorem
on the role of the relative entropy distance in resource theories.

\section{Purity theory and thermodynamics}
\label{sec:purity}

Purity theory \cite{uniqueinfo} is perhaps one of the simplest of resource theories, and is intimately linked with both thermodynamics
and quantumness of correlations.  We shall therefore briefly review its structure.  In essence, it's the study of entropy, and thus it 
can be viewed as thermodynamics where the Hamiltonian $H=0$.  The class of operations $\classops$ is noisy operations (NO) which consists of 
\bei
\item adding systems in the maximally mixed state $\id/d$
\item partial trace
\item unitary transformation
\eei
The only state which can be created under this class of operations is the maximally mixed state, and so, any other state is a resource.  In fact,
the more pure a state is (in terms of having low entropy), the more valuable it is.  The theory is reversible, in the sense that if one can
transform $\rho^{\otimes n}$ into $\sigma^{\otimes m}$, then one can also transform $\sigma^{\otimes m}$ into $\rho^{\otimes n}$.  As a result, using
a general theorem which we will describe in Section \ref{sec:rel}, the number of copies of $\sigma$ which can be formed per input copy of
$\rho$ i.e. $m/n$ 
is given by $(\log d - S(\sigma))/(\log d - S(\rho))$, and the {\it negentropy}  $\log d - S(\rho)$ is the unique measure of purity.  Since the
negentropy of a pure state is just $\log d$, we can view the negentropy of any other state, as giving the number of pure qubits which can be
extracted, or distilled from a given state $\rho$.  Because the theory is reversible, this is also the number of pure qubits needed to create
the state.  The transition from some mixed state into a pure state, if done in the presence
of a heat bath at temperature $T$, is sometimes called Landauer's erasure, 
and the amount of work required to make this transition is just $kT$ times the necessary
change in negentropy of the state, with $k$ Boltzmann's constant. Thus there is an intimate connection between purity theory, and thermodynamics.  In fact,
as we have stated, purity theory is just thermodynamics, but with $H=0$.

One can also characterise state transformation in the case where one has a single system, rather than a large number of copies of a state.
A necessary and sufficient condition for a transition is given by the majorization condition as shown in \cite{uniqueinfo}, 
and this now called the single-shot scenario.  The case where one allows
a small probability of failure in the single shot case, was given in \cite{dahlsten2011inadequacy}. 

If we now add a Hamiltonian, and must preserve energy, then we have the resource theory of thermodynamics\cite{janzing_thermodynamic_2000,thermoiid}.  
The theory 
is still reversible, and rather than the
negentropy governing state transitions, one then finds that the free energy $F=\tr H\rho - TS(\rho)$ governs state transitions.  The single-shot
case of thermodynamics was considered in \cite{horodecki-singleshot,aberg_truly_2011}.  This is important if we wish to consider the thermodynamics of single systems,
or highly correlated ones.  There one finds that the theory is not reversible, and thus the free energy does not govern state transitions, indeed there are
at least two free energies\cite{horodecki-singleshot}, one which is analogous to the distillable entanglement and one which is analogous to the entanglement of formation (they give
the amount of work which can be extracted from a state, and the amount of work which is needed to perform the reverse process when created a state).

{\it Digression: noisy operations and Asymptotic Birkhoff Conjecture.}
As an aside, it is interesting to see that noisy operations have been recently used to disprove the asymptotic Birkhoff conjecture\cite{birkhoff-conjecture}. 
The conjecture was that the $n$-fold tensor product of a  bistochastic map
can be approximated by maps that are mixtures of unitaries in the limit of large $n$ \cite{svw2005}. 
In \cite{Haagerup-Musat2010} the question was posed for a larger class of operations, that turned out to be 
just the noisy operations \footnote{In the paper by \cite{Haagerup-Musat2010}  a class called factorisable maps was used. 
Andreas Winter (private communication) has noted that the one can use a class of noisy operations,
which is a subclass of factorisable maps (perhaps it is just the same class).}. 
It turned out that then the question easily reduces 
to a question of equality of bistochastic maps and noisy operations on a fixed system (without tensoring).
In \cite{Haagerup-Musat2010} it was shown that it is not true, by showing that noisy operation 
is extremal in the set of trace preserving maps only when it is unitary, 
and at the same time presenting bistochastic maps that are extremal. 
The proof that noisy operations are generically non-extremal is quite involved exploiting tools from 
the field of von Neumann algebras. One expects that a simpler proof can be found.

\section{Meaning of the relative entropy distance}
\label{sec:rel}

Here we will discuss the special role in resource theories, of a quantity known as the relative entropy distance.  
A theory is reversible, if one can go from any resource state to any other state, 
in a reversible manner.
More precisely, for any states $\rho,\sigma$ 
if one can transform $n$ copies of $\rho$ into $m$ copies of $\sigma$ 
in the asymptotic limit of large $m,n$, then one can also transform back the $m$ copies of $\sigma$ 
into $n$ copies of $\rho$. Furthermore, given any two resource states $\rho,\sigma$, there exists some $n,m$ such that
the transition is possible.
In a reversible theory, there is a unique measure on resources, as can be seen by an analogous argument to that used to show
the optimality of the Carnot cycle~\cite{popescu-rohrlich}: roughly speaking, a quantity, that can only decrease under transformations, 
must be unique. 
Now in \cite{thermo-ent2002} it it was pointed out  that this unique quantity must be the relative entropy distance 
from the set of free states  (in a regularized form). 
Before stating the result more precisely, let us set some notation. 

Systems are determined by Hilbert spaces.   Two systems give rise to a larger system, described by the tensor product of the two Hilbert spaces
\footnote{We can have a theory where a basic system is  composite (e.g. bipartite systems). If the original systems were bipartite, the bipartite structure is inherited in a natural way 
in a new system (same for the multipartite case): if for the first system we have $\hcal_A \ot \hcal_B$ 
and for the second one $\hcal_{A'} \ot \hcal_{B'}$  
then the new system is given by $ (\hcal_A\ot \hcal_{A'}) \ot (\hcal_{B}\ot\hcal_{B'}$.}.
For each system we have the set of states acting on its Hilbert space. We distinguish the set $\useless$ of free states,
that is closed under composition of systems described above, i.e. it is closed under tensor product. 
In entanglement theory this is e.g. the set of separable states. But also this is true of the set of PPT states,
or the set of bound entangled states.  We recall that these are called {\it bound sets $\boundset$}.

Our result emphasizing the role of the relative entropy distance will require convexity of this bound set. 
Therefore, the set of bound entangled states cannot be used, provided that NPT bound entangled states 
exist (which is a long standing open problem\cite{horodecki_quantum_2009}, see \cite{npt-few-steps,npt-k-extendible} for recent developments). 


We will present here the following theorem which is a slightly modified version of one given in \cite{thermo-ent2002}. We will later state and prove
a stronger but more technical version of it.  Recall
that the relative entropy is defined as
\be
S(\rho|\sigma)=\tr\rho\log\rho-\tr\rho\log\sigma
\ee
and that the relative entropy distance to a set $\useless$ is
\be
E_r(\rho)=\inf_{\sigma\in\useless}S(\rho|\sigma)
\ee
and the regularisation of this quantity is simply
\be
E_r^\infty(\rho)=\lim_{n\rightarrow\infty}\frac{E_r(\rho^{\otimes n})}{n}.
\ee
Recall also that the rate $R(\rho\to \sigma)$ of transitions between $\rho$ and $\sigma$ is
the number of copies of $\sigma$ which can be produced per input copy $\rho$ in the limit of many copies.
\begin{theorem}
\label{thm:hyperset}
Consider an arbitrary class of operations $\classops$, being a subclass of completely positive maps and assume that
\bei
\item there exists a convex set of states $\useless$, which contains the maximally mixed state 
and is closed under the class of operations.
\item the theory is asymptotically reversible (i.e. all states that lie outside of $\useless$
are asymptotically interconvertible.
\item there exists state $\rho_0$, for which regularized relative entropy distance from the set $\useless$
denoted by $E_r$ is nonzero.
\eei
Then for any two states $\rho$ and $\sigma$ outside of $\useless$, the asymptotic conversion rate is given by the ratio of regularized relative entropy distance from the set:
\be
R(\rho\to \sigma)= \frac{E_r^\infty(\rho)}{E_r^\infty(\sigma)}
\label{eq:relative-rate}
\ee
\end{theorem}
In reversible theories one can usually choose a reference state $\rho_0$,  
for which $E_r(\rho_0)=1$. Then, the theorem alternatively says that for any state $\rho$ 
we have 
\be
R(\rho\to \rho_0)= E_r(\rho)
\ee
The theorem shows that the relative entropy distance is a distinguished measure of resource as it is unique 
measure in reversible theories.  
{\bf Remarks:} As we will discuss at the end of this Section, a stronger version of the theorem removes the
necessity for the set $\useless$ to contain the maximally mixed state.  It is proven in the Appendix. 
As was noted in \cite{Gour-rel-frame2009}, the theorem needs to assume that $E_r^\infty$ is nonvanishing on resource states,
something which is not satisfied in asymmetry theory. 

Three examples of reversible theories illustrate the utility of the above theorem: purity theory, 
the much more involved example - entanglement theory with a special class of operations, and thermodynamics:
\bei
\item {\bf Purity theory} \cite{uniqueinfo}: 
\be
R(\rho\to |0\>\<0|) = \log d -S(\rho) 
\ee
where $|0\>\<0|$  is a pure state of a qubit.  The above quantity is precisely equal to $E_r(\rho|I/d)$


\item {\bf Entanglement theory} \cite{BrandaoPlenio2007-separable-maps}
\be
R(\rho\to\psi_+)=E_r^\infty(\rho)
\ee
where $E_r$ is the regularized relative entropy distance to the set of separable states, 
and the class of operations is the set of all $\ep$-non-entangling maps. (Ideally one would like to have here 
the set of non-entangling maps, i.e. maps that do not produce entanglement, however due to some technical difficulties 
the authors used maps that can produce $\epsilon$ amount of entanglement, provided, one takes eventually a limit $\ep\to 0$).

\item {\bf Thermodynamics} (of states commuting with a Hamiltonian $H$): \cite{thermoiid}
\be
R(\rho\to |E\>\<E|)= \frac{F(\rho)- F(\rho_\beta)}{F(E)- F(\sigma_\beta)} = \frac{E_r(\rho||\rho_\beta)}{E_r(|E\>\<E|||\sigma_\beta)}
\ee
where $\rho_\beta,\sigma_\beta $ are the  Gibbs states of the Hamiltonian of the initial and final system 
respectively, and the Free Energy is given by $F= \tr H\rho-TS(\rho)$ with $T$ the temperature -- the free energy is just the relative
entropy distance to the thermal state \cite{donald1987free}.  Amusingly, in the single shot case, both the work obtainable from 
a state and the work required to perform the reverse procedure, are given by generalisations of the relative entropy distance
to the thermal state~\cite{horodecki-singleshot}. 
\eei

The first two examples fit directly into this theorem, while the connection between the third - thermodynamics 
and our theorem we will discuss later. But, first let us note 
an interesting twist here. Usually, one defines a resource theory by choosing a class $\classops$ of operations. Most often this naturally induces 
a set $\useless$ of states which could be created under that class from some initially chosen state. 
One can now define a generally larger and more powerful 
class of operations, an expanded class $\classops'$, which are all the operations which when given states from $\useless$, do not create any states outside $\useless$.
For example, in entanglement theory, the class of LOCC induces the class of separable states which can be created under that class.  And this in turn
induces another class of operations, namely, all non-entangling operations which map separable states to separable states. It turns out that the above result on reversibility of 
entanglement theory, can be stated much more generally (it was announced in \cite{BrandaoPlenio2007-separable-maps} and proved in  \cite{Fernando-unpublished-reversibility}): 

{\bf General reversibility of resource theories:}
{\it Consider a resource theory, for which any resource state possesses 
the so called exponential distinguishing property relative to the set $\useless$.
Such theory is reversible  under the class of operations, that asymptotically do not create resource states from  $\useless$. The interconversion rate is determined by the relative entropy distance from $\useless$
as in \eqref{eq:relative-rate}.}

Let us finally discuss thermodynamics. This theory does not satisfy one condition, namely,
the set $\useless$ consists only of the Gibbs state, and therefore does not contain the maximally mixed state. 
However, one can construct a version of the theorem, 
which would sound more technical, yet covers this case too. 
Namely, the maximally mixed state was needed to ensure that for any state, $E_r$ 
is bounded by $\log d$ times some constant, which allows one to prove asymptotic continuity 
of $E_r$.  Now from the proof it follows that the only thing we  need is that $E_r$ is asymptotically continuous  on many copies. I.e., the rising dimension comes only from the fact that the number of copies of the system increases. 
The state of many copies of the system may be global, but the Hamiltonian is separate 
for each copy. Therefore the single free state is a tensor product of many copies of the Gibbs states 
of a fixed dimension. In such case, the required boundedness of $E_r$ follows from subadditivity of entropy.

\section{Deficit, discord, and relative entropy of quantumness}
\label{sec:deficit}

Let us now review the notion of quantumness of correlations, and the deficit.  We  will then, in Section \ref{sec:def-resource},
relate this to the resource theory of purity.
The quantum deficit and quantum discord are closely related quantities, though they have risen from two separate philosophies. 
Quantum discord was designed to distinguish between quantum correlations and classical correlations. 
Quantum deficit was introduced to quantify the amount of thermodynamical work which is lost due to the separation of  two subsystems. It turned out, that this loss is related not to entanglement but to
quantumness of correlations, and in this way, it become a measure of quantumness.

{\it Basic versions.}
It is easiest to compare the one-way quantum deficit with the quantum discord. The quantum one-way deficit $\Delta^\rightarrow$
is the minimal difference of entropies between the original state, and a state subjected to 
a complete von Neumann measurement on the first subsystem, while the quantum discord $\delta$ (if we optimise the quantity over measurements) is the minimal difference of mutual informations 
under the same conditions:
\ben
&\delta(\rho_{AB})=  \min_{{P_i}}\left(I_M(\rho_{AB})-I_M(\rho'_{AB})\right) \nonumber \\
&  \Delta^\to(\rho_{AB})= \min_{{P_i}}\left(S(\rho'_{AB}) - S(\rho_{AB})\right)
\label{eq:def-disc}
\een
where the minimum is taken over all von Neumann measurements  $\{P_i\}$ with $P_i$ rank one projectors,
and 
\be
\rho'_{AB}=\sum_i P_i \ot I \rho_{AB} P_i \ot I.
\ee
If one writes explicitly the mutual information in terms of entropies, one realizes
that the difference between those two quantities is due to entropy production on Alice's side, 
which is counted in the case of deficit, but not in the case of discord. 
\ben
& \delta(\{P_i\}) = [S(\rho_{AB}')-S(\rho_{AB}) ] - [S(\rho_A')-S(\rho_A)]\nonumber\\
& \Delta^\to(\{P_i\}) = [S(\rho_{AB}')-S(\rho_{AB}) ]
\een
where $\delta(\{P_i\}), \Delta^\to(\{P_i\})$ are defined as in \eqref{eq:def-disc}
but for a given measurement (i.e.  without minimization over measurements). 

{\it Remark.} Both quantities may depend not only on the state itself, but also 
on the Hilbert space on which it is supported. 

{\it Extended versions.}
One can allow more general versions of both quantities as follows. 
For quantum discord, one may take infimum over POVM's rather than over von Neumann measurements:
\be
\delta_{gen}= I_M(\rho_{AB}) - I_{cl}(\rho_{AB})
\label{eq:gen-disc}
\ee
where 
\be
I_{cl}(\rho_{AB})=\max_{\{A_i\}} I_M(\sum_i p_i |i\>\<i|\ot \rho^i_B)
\ee
where $\rho^i_B=\tr_A (A_i \ot I \rho_{AB} A_i \ot I)$ are Bob's postmeasurement states,
and $p_i=\tr A_i^\dagger A_i \rho_{AB}$ are the probabilities of Alice's outcomes.

Regarding the quantum deficit, we cannot use POVM's, as they hide the entropy production in the reference system
used to implement the von Neumann measurement. Rather, we embed the state into a larger system on Alice side
(adding some extra dimensions), and compute the deficit on such a state. 
\be
\Delta_{gen}^\to(\rho_{AB})= \inf_{V} \Delta^\to(\tilde\rho_{A'B})
\ee
where $\tilde\rho_{AB} =V\ot I_B \rho_{A'B} V^\dagger \ot I$, where $V$ is a local isometry $V:\hcal_A\to \hcal _A'$.

One can find, that the definition of generalized discord  of \eqref{eq:gen-disc} can be equivalently written
in the same way:
\be
\delta_{gen}(\rho_{AB})=\min_V \delta(\rho_{A'B})
\ee
where again the infimum is taken over all local isometries $V$.

{\it Regularizaton: one-way deficit meets discord.}
One may define the regularization of a function $f$ as 
\be
f^\infty(\rho)=\lim_{n\to \infty} \frac1n f(\rho^{\ot n})
\ee
whenever the limit exists. It turns out that after regularization, 
the discord optimised over measurements and the one-way deficit become the same, as proven by Igor Devetak \cite{igor-deficit}.
Moreover, the difference between the generalized quantities and the basic versions disappears in such limit:
\be
\delta_{gen}^\infty=\delta^\infty=(\Delta^{\to})^\infty=(\Delta_{gen}^{\to})^\infty
\ee

{\it Zero way and two way deficits.}
The quantum deficits are defined operationally, i.e. we count the amount of work 
that can be obtained by means of a state, and classical communication. 
The one-way deficit
is related to the amount of work when one-way classical communication is allowed. 
Analogously, the zero-way deficit is defined when communication is only allowed after both parties perform measurements,
and two-way deficit - when two-way classical communication is allowed. 

As pointed out in \cite{huge-delta} it is easy to cook up a zero-way version of discord:
\be
\delta^\emptyset(\rho)= \inf_{\{P_i\},\{Q_j\}} \left(I_M(\rho_{AB}) - I_M(\rho'_{AB})\right)
\ee
where $\rho'_{AB}= \sum_{ij} P_i\ot Q_j  \rho  P_i\ot Q_j $,
and $\{P_i\},\{Q_j\}$ are complete von Neumann measurements, $P_i$ and $Q_j$ 
are rank one projectors. However it is an open question how to define a sensible notion of two-way discord.   

{\it Deficit and  relative entropy distance.}
In \cite{huge-delta} the one-way and zero-way discord have been shown to be equal to the relative entropy distance 
to two classes of states: so called c-q and c-c states, respectively.  
The c-q states are all states of the form 
\be
\sum_i p_i |\psi_i\>\<\psi_i| \ot \rho_i
\ee
where $\psi_i$ are orthonormal states. I.e. the first subsystem is classical. 
The c-c state is of the form 
\be
\sum_{ij} p_{ij} |\psi\>\<\psi_i| \ot |\phi_j\>\<\phi_j|
\ee
where $p_{ij}$  is a probability distribution, and $\{\psi_i\}$, $\{\phi_j\}$ are 
orthonormal bases. I.e. both systems are classical.

There is a long standing open problem, of whether one can represent two-way quantum deficit in this way.
The expected class of states is the one called pseudo-classically correlated states, 
see \cite{huge-delta} for a definition. The problem is roughly, whether in the process of making a state 
c-c with minimal entropy by local dephasing (i.e. measurements without readout) and communication, Alice (Bob) 
can restrict themselves to only making finer and finer measurements at any round, 
or whether at some round, Alice (Bob) will need to make a measurement, which does not commute
with the previous measurements.

\section{Relation between purity theory and deficit}
\label{sec:def-resource}

The theory of quantumness is not, at least at the present stage of development, 
a resource theory. The main problem is that we do not have here a class of operations, 
that would not increase the quantumness. Quantumness can be created by local operations and 
classical communication, and even by local operations, by simply tracing out a system 
containing orthogonal flags correlated with nonorthogonal states \cite{Dagmar-quantumness-increase,Ciccarello-discord-increase}.

However, at least the quantum deficit is strictly related to a resource theory  -
the theory of {\it local} purity \cite{nlocc}.  The key element we need is that because quantumness of correlations requires
 distinguishing classically correlated states from other sorts of correlations, the resource theory should expose this
distinction.  Local purity theory is a non-reversible theory,
which arises, when we restrict the class of operations in purity theory. 
Namely we consider a class of noisy-LOCC operations (NLOCC) consisting of 
\bei
\item local noisy operations (NO) as described in Section \ref{sec:purity}
\item sending qubits via dephasing channels
\eei
i.e. Alice and Bob can each perform NO operations locally, and between them, they have a dephasing channel (which is a channel
which formalises the notion of classical communication).

As with purity theory, we are interested in state transitions in general, but in particular, the number of pure states that can be extracted
from a state.  Here, we are interested in extracting the purity locally (this is in a very real sense, complementary\cite{compl} to extracting
non-local entanglement).  We thus define the distillable local purity  as the maximal rate of qubits in a pure product state  that can be obtained 
from many copies of a given state $\rho_{AB}$. Let us call the quantity $I_l^{NLOCC}$ which can be called also the {\it localisable information}, since purity can be treated 
as information \cite{nlocc}). The subscript $l$ stands for 'localisable'.

The distillable purity depends on the Hilbert space, since embedding a state into 
a larger Hilbert space is the same as increasing local purity.
For example, an embedding is $\rho_{AB} \to |0\>_{A'}\<0| \ot \rho_{AB}$ 
where clearly the distillable local purity is increased at least by 1 (due to one additional pure local qubit).

One can also consider the same scenario, yet with a slightly different class of operations, called closed-LOCC operations (CLOCC) instead of noisy-LOCC operations.
The important thing these two classes of operations have in common is that they don't allow pure states for free, and thus low entropy becomes a resource. 
Namely, closed-LOCC operations are given by:
\bei
\item local unitary transformations
\item local partial trace
\item  sending qubits via dephasing channels
\eei
The resulting distillable local purity will be denoted by $I_l^{CLOCC}$.

The class of operations NLOCC is more tractable technically than CLOCC.
The class CLOCC does not allow the addition of any ancillas whatsoever, which automatically
imposes non-creation of purity. The NLOCC class allows for preparation of maximally mixed states only. 
It is not known, whether the distillable local purity under CLOCC is equal to that under NLOCC, but we conjecture that it is so. 

A one-way version of distillable local purity denoted by $I_l^{CLOCC,\to}$  or $I_l^{NLOCC,\to}$
is when one allows only one way communication, say from Alice to Bob. Then we ask about the rate of 
purity on Bob's side (even if Alice will obtain some qubits on her side, she can send them for free
to Bob down the dephasing channel).   

The quantities $I_l$ are monotones under CLOCC/NLOCC, which is a feature needed 
in a reasonable resource theory.  $I_l$ is therefore a measure of local purity.

Now, the relation between the quantum deficit $\Delta$ and distillable local purity by CLOCC operations is straightforward.
Namely 
\ben
&\Delta^\infty(\rho_{AB}) = N-S(\rho_{AB}) - I_l^{CLOCC}\\ \nonumber
&(\Delta^\to)^\infty(\rho_{AB}) = N-S(\rho_{AB}) - I_l^{CLOCC,\to}\\ \nonumber  
\een
where $N=\log (|A||B|)$ is the total number of qubits, and the $\infty$ superscript denotes the fact that we compute these quantities in the limit of many copies.
The deficit can thus be interpreted in terms of local purity: 
it is the difference between the local purity that can be obtained, 
if quantum communication is allowed (given by $N-S$)
and the local purity that can be obtained, 
if communication is restricted to the dephasing channel (one- or two-way).  Note that comparison between the amount of purity which can be distilled
under CLOCC and under closed local operations, allow one to define a measure for classical correlations~\cite{compl,huge-delta}.
We thus see that the deficit, and hence quantumness, is very closely related to the resource theory of local purity.

\section{Conclusion and open problems}

We have discussed resource theories in general, as well as a number of examples.  In particular, we have discussed the notion of quantumness of correlations 
in the context of resource theories.  This provides a natural setting in which the discord, deficit and other related quantities find an operational
meaning.  There are many open questions, some of which were raised in \cite{huge-delta}.  For example, it would be interesting to know
 whether one can represent the two-way quantum deficit as a relative entropy distance.  This relates to the problem of whether, when extracting local purity,
all the operations should commute or not.  Another open problem is to prove (or disprove) the conjecture that 
the distillable local purity under CLOCC is equal to that under NLOCC. An interesting open question 
concerns the relation between the discord and deficit. 
Originally, the discord was a quantity related to one-way communication, 
and its natural analogue arising in local purity theory was the one-way deficit. 
It was easy to define the zero-way discord - as a counterpart to the zero-way deficit.
However so far there is no two-way version of discord that would correspond to the two-way deficit.

There are also a number of interesting questions related to resource theories in general.  For example, we know some situations where the relative entropy
distance is the unique measure of the resource, but it would be interesting to find other conditions when this is so.  It would also be interesting
to find more examples of resource theories, so that we can compare the different structures.  For example, resource theories where there are many bound
sets are of interest, as are hybrid resource theories, which can be constructed by combining resource theories.  Finding commonality across various
resource theories are of interest.  For example, the notion of entanglement spread~\cite{Har09} as a resource and the use
of embezzling states~\cite{ent-embezzling} in communication tasks is
virtually identical to the use of reference frames which allow one to lift a superselection rule in asymmetry theory~\cite{thermoiid}.
We believe that future study of resource theories will lead to general results with broad applicability.

{\bf Acknowledgements.}
MH thanks the support of the Polish Ministry of Science and Higher Education
Grant no. IdP2011 000361. JO is grateful to the Royal Society for their support.  Part of this work was done at National
Quantum Information Centre of Gdansk.

\section{Appendix}
Here we give formal definition of several notions, and provide a proof for theorem \ref{thm:hyperset}.
\begin{definition}[Asymptotically continuous function.]
A function  $f$  is asymptotically continuous  if for family of Hilbert spaces $\hcal_n$ 
with $\dim(\hcal_n)\to \infty$ and states $\rho_n,\sigma_n$ such that $||\rho_n-\sigma_n||_1\to 0$ 
we have 
\be
\lim_{n\to \infty}\frac{|f(\rho_n)-f(\sigma_n)|}{\dim\hcal_n}=0
\ee
\end{definition}

\begin{definition}[Asymptotic rate of transition.]
Asymptotic rate of transition between two states $\rho$ and $\sigma$ 
is defined as follows. Suppose that there exists a protocol $\pcal$ i.e. 
a sequence $m_n$, and a sequence of operations $\pcal=\{\Lambda_n\}$ such that 
\be
\lim_{n\to \infty} ||\Lambda_n(\rho^{\ot n}- \sigma^{\ot m_n}||_1 =0.
\ee
Then we define rate of the protocol as 
\be
R(\pcal)=\limsup\frac{m_n}{n}.
\ee
The asymptotic rate of transition is then given  by
\be
R(\rho\to \sigma)=\sup_\pcal R(\pcal).
\ee
\end{definition}

\begin{definition}[Asymptotically reversible resource theory.]
A theory given by a class of operations $\classops$ (subset of completely positive maps,
and a set of states $\useless$ closed under the class, 
is reversible, when for any two states $\rho$ and $\sigma$ outside of 
the class we have  $0<R(\rho\to \sigma)<\infty$ and 
\be
R(\rho\to \sigma)R(\sigma\to \rho)=1.
\ee
 \end{definition}
\begin{proof}[Proof of theorem \ref{thm:hyperset}]
As proved in  \cite{Michal2001} (Theorem 4) if a function $f$ is asymptotically continuous, 
monotonous under class of operations,
its regularization exists, and $R(\rho\to \sigma)<\infty$ then 
\be
f^\infty(\sigma) R(\rho\to \sigma) \leq f^\infty(\rho).
\ee
Let us take now a state $\rho_0$ for which $f(\rho_0)>0$.
We then  have 
\be
\frac {f^\infty(\rho)}{f^\infty(\rho_0)}\leq R(\rho\to \rho_0), 
\label{eq:mon-rate}
\ee
and also 
\be
f^\infty(\rho)R(\rho_0\to \rho)\leq f^\infty(\rho_0).
\ee
Multiplying the latter with $R(\rho\to \rho_0)$ and using reversibility, 
we obtain  an inequality converse to \eqref{eq:mon-rate},
which gives 
\be
\frac {f^\infty(\rho)}{f^\infty(\rho_0)}= R(\rho\to \rho_0).
\ee
This implies that $f^\infty(\rho)>0$ for any state outside of $\useless$.
Applying now the above equality to state $\sigma $ instead of $\rho$ 
multiplying the equalities, and using again reversibility, we obtain 
\be
R(\rho\to \sigma)=\frac{f^\infty(\rho)}{f^\infty(\sigma)}.
\ee 
Finally, note that, as proven in \cite{Synak05-asym}, under the above assumptions on 
the set $\useless$, relative entropy distance $E_r$ from this set is asymptotically continuous. 
Therefore we obtain 
\be
R(\rho\to \sigma)=\frac{E_r^\infty(\rho)}{E_r^\infty(\sigma)}.
\ee 
\end{proof}

\bibliography{ref-resource,../refjono2,../refjono,../rmp12,../work,../refthermo}

\end{document}